\documentclass[12pt]{amsart}
\usepackage{tikz-cd}
\usepackage{amssymb}
\usepackage{amsmath}
\usepackage{graphicx,color}
\usepackage{wrapfig,framed}
\usepackage[height=22cm, width=16cm, hmarginratio={1:1}]{geometry}
\usepackage{hyperref}
\usetikzlibrary{decorations.markings}
\tikzset{double line with arrow/.style args={#1,#2}{decorate,decoration={markings,%
mark=at position 0 with {\coordinate (ta-base-1) at (0,1pt);
\coordinate (ta-base-2) at (0,-1pt);},
mark=at position 1 with {\draw[#1] (ta-base-1) -- (0,1pt);
\draw[#2] (ta-base-2) -- (0,-1pt);
}}}}

\theoremstyle{plain}
\newtheorem{theorem}{Theorem}

\newtheorem{lemma}{Lemma}

\theoremstyle{definition}

\newtheorem{remark}{Remark}

\newcommand{\beq}{\begin{equation}}
\newcommand{\eeq}{\end{equation}}
\newcommand{\nn}{\nonumber}
\newcommand{\F}{\mathcal{F}}
\newcommand{\A}{\mathcal{A}}

\newcommand{\ZZ}{{\mathbb Z}}
\newcommand{\tr}{{\rm tr}}
\newcommand{\e}{\epsilon}
\newcommand{\p}{\partial}

\begin{document}

\title{KdV integrability in GUE correlators}
\author{Di Yang}
\address{School of Mathematical Sciences, University of Science and Technology of China, Hefei 230026, P.R.~China}
\email{diyang@ustc.edu.cn}
\date{}
\begin{abstract}
Okounkov~\cite{O00} proved a remarkable formula relating 
$n$-point GUE (Gaussian unitary ensemble) correlators of a fixed genus to  
Witten's intersection numbers of the same genus.
The partition function of GUE correlators is a tau-function for the Toda lattice hierarchy. In this note, 
based on the knowledge of these two statements we give a new proof of the Witten--Kontsevich theorem, 
 that relates Witten's intersection numbers to the KdV (Korteweg--de Vries) integrable hierarchy.  
\end{abstract}

\maketitle

\setcounter{tocdepth}{1}
\tableofcontents

\section{Introduction}\label{section1}
The Korteweg--de Vries (KdV) equation
\beq\label{KdVeq}
\frac{\p u}{\p t_1} = u \frac{\p u}{\p t_0} + \frac1{12} \frac{\p^3 u}{\p t_0^3}
\eeq
 is a nonlinear evolutionary partial differential equation (PDE).
Since the late 1960s, it has been known from the theory of integrable systems that the KdV
 equation~\eqref{KdVeq} can be extended to an infinite system of pairwise commuting evolutionary PDEs, 
called the {\it KdV hierarchy}, which can be written compactly as follows:
\beq\label{KdVhie}
\frac{\p u}{\p t_d} = \frac{1}{(2d+1)!!} \bigl[\bigl(L^{\frac{2d+1}{2}}\bigr)_+,L\bigr], \quad d\ge1.
\eeq
Here, $L:=\p_{t_0}^2 + 2 u$ is called the Lax operator, $[\,,\,]$ denotes the commutator, 
and for a pseudodifferential operator $P$ the notation $P_+$ means taking the nonnegative 
part of~$P$ (see e.g.~\cite{Dickey03} for more details).
The $d=1$ equation in~\eqref{KdVhie} coincides with~\eqref{KdVeq}. 
Note that the normalization of KdV flows~\eqref{KdVhie} differs from~\cite{Dickey03} by rescalings.

Let $g,n$ be nonnegative integers satisfying the stability condition $2g-2+n>0$, and $\overline{\mathcal M}_{g,n}$ the Deligne--Mumford 
moduli space of stable algebraic curves of genus~$g$ with $n$~distinct marked points.
Denote by $\psi_a$ the first Chern class of the $a$th tautological line bundle on $\overline{\mathcal M}_{g,n}$, $a=1,\dots,n$.
They are often called {\it psi-classes}. The integrals 
\beq\label{psiintegralsWnotation}
 \int_{\overline{\mathcal M}_{g,n}} \psi_1^{d_1}\cdots \psi_n^{d_n} = :\langle \tau_{d_1} \dots \tau_{d_n}\rangle_g
\eeq
are called {\it psi-class intersection numbers} or {\it Witten's intersection numbers}, which
vanish unless the degree and the dimension match:
$d_1+\dots+d_n= 3g-3+n$. Here, \hbox{$d_1,\dots,d_n\ge0$}. The integrals~\eqref{psiintegralsWnotation} will be understood as~0 
if $2g-2+n\le0$. Let ${\bf t}=(t_d)_{d\ge0}$ be indeterminates.
The {\it Witten free energy} $\F=\F({\bf t})$ is the power series of~${\bf t}$ defined by
\beq\label{Wittenfreeenergy}
\F({\bf t})=\sum_{g\geq0} \sum_{n\geq 0} 
\sum_{d_1, \dots, d_n\geq 0} \frac{\langle \tau_{d_1} \dots \tau_{d_n}\rangle_g}{n!} t_{d_1}\cdots t_{d_n}.
\eeq

In~\cite{Wi91} Witten proposed a striking conjecture:

\smallskip

\noindent {\bf Witten's conjecture} (\cite{Wi91}). 
{\it The function $u := \p^2 \F/\p t_0^2$ obeys the KdV hierarchy~\eqref{KdVhie}.}

\smallskip 

Witten's conjecture was first proved by Kontsevich~\cite{Ko92}, and is now 
 known as the {\it Witten--Kontsevich theorem}. We refer to~\cite{AIS21, Bur17, CLL08, ELSV01, KL07, KimLiu09, Mir07, O00, O02, OP09} 
 about other proofs.
 Note that, to prove the Witten--Kontsevich theorem, 
 it suffices~\cite{KL07, LX11-1} to prove the $d=1$ case of~\eqref{KdVhie}, i.e., the KdV equation for psi-class intersection numbers. Indeed,
Witten proved~\cite{Wi91} that psi-class intersection numbers obey the following two equations
\begin{align}
&\langle\tau_0\tau_{d_1}\dots\tau_{d_n}\rangle_g=\sum_{1\le a \le n,\, d_a>0} 
\langle \tau_{d_a-1} \prod_{b\neq a} \tau_{d_b} \rangle_g + \delta_{n,2}\delta_{g,0}, \label{stringeqint}\\
&\langle\tau_1 \tau_{d_1}\dots\tau_{d_n}\rangle_g= (2g-2+n)\langle \tau_{d_1}\dots\tau_{d_n}\rangle_g + \frac1{24} \delta_{g,1} \delta_{n,0},\label{dilatoneqint}
\end{align}
called {\it string equation} and {\it dilaton equation}, respectively, and, according to~\cite{KL07, LX11-1}, 
validity of \eqref{stringeqint}, \eqref{dilatoneqint} and the KdV equation for intersection numbers implies validity of~\eqref{KdVhie}.

In Kontsevich's proof~\cite{Ko92} of Witten's conjecture, a remarkable identity was established:
\beq\label{Kontsmi}
\sum_{d_1,\dots,d_n\ge0} \langle \tau_{d_1}\dots \tau_{d_n}\rangle_g \prod_{a=1}^n \frac{(2d_a-1)!!}{z_a^{2d_a+1}} 
= \sum_{G\in G_{g,n}^3} \frac{2^{2g-2+n}}{|{\rm Aut} (G)|} \prod_{e\in E(G)} \frac{1}{\tilde z(e)}
\eeq
(see also~\cite{OP09}).
 Here, $g,n\ge0$ satisfy $2g-2+n>0$, $G_{g,n}^3$ denotes the set of trivalent maps on an oriented closed surface of 
 genus~$g$ with $n$ cells marked by $1,\dots,n$, 
 $\rm {Aut}(G)$ is the finite group of symmetries 
 generated by orientation preserving homeomorphisms of the surface
 that map $G$ to~$G$ and respect the markings, $E(G)$ denotes the edge set of~$G$, 
 and the notation $\tilde z(e)$ is defined as follows:
let the variables $z_1,\dots,z_n$ correspond to the markings of $G\in G_{g,n}^3$; each edge $e\in E$ borders two cells; let $i$ and $j$ be the labels assigned to these cells (if both sides of~$e$ border the same cell then $i = j$); $\tilde z(e):=z_i+z_j$.

We call~\eqref{Kontsmi} the {\it Kontsevich main identity}. 
This identity connects psi-class intersection numbers to combinatorics of trivalent maps.
In order to obtain the KdV integrability, Kontsevich~\cite{Ko92}  
constructed the ``matrix Airy function" (the 
Kontsevich matrix model) and completed his proof of Witten's conjecture. 
Another proof of the KdV integrability, again after the Kontsevich main identity~\eqref{Kontsmi} is established,
is given by \hbox{Okounkov~\cite{O02}} 
using a limit formula (see~\cite[(2.7) and Proposition~1]{O00}, \cite{OP09} or~\eqref{Ok} below) and the edge-of-the-spectrum model~\cite{O00, O02, OP09}.

Enumeration of maps, or say ribbon graphs, appeared naturally in 
the GUE (Gaussian unitary ensemble) random matrix model~\cite{BIZ80, DFGZJ95, HZ86, OP09, Wi91},
which, by the theory of orthogonal polynomials, is well known to be
related to the Toda lattice hierarchy ({\it aka} the 1D Toda lattice or the Toda chain), 
and more specifically to the Volterra lattice hierarchy ({\it aka} the discrete KdV flows) when restricted 
to even couplings; see~e.g. \cite{Deift99, DFGZJ95, DLYZ20, DY17, DY17-2, GMMMO91, Mar91, Wi91}. 
The goal of this paper is to prove Witten's KdV equation from Toda integrability based on the 
above-mentioned limit formula of~\cite{O00}. We note that
the KdV integrability was discovered by taking double-scaling/continuum limit in matrix gravity (cf.~\cite{BDSS90, BK90, DFGZJ95, D90, 
DS90, GMMMO91, GM90, M94, Wi91}) and our proof resembles this original idea.

Let $N$ be a positive integer and $\mathcal{H}(N)$ 
 the space of hermitian matrices of the finite size~$N$. 
By {\it GUE correlators} we mean the normalized Gaussian integrals (cf.~\cite{BIZ80, BIPZ78, DFGZJ95, HZ86})
\beq\label{Gaussianintegral}
\frac{\int_{\mathcal{H}(N)} \tr (M^{i_1}) \dots \tr (M^{i_n}) e^{-\frac1{2} \tr (M^2)} dM}{\int_{\mathcal{H}(N)} e^{-\frac12 \tr (M^2)} dM} =: 
\langle \tr M^{i_1} \dots \tr M^{i_n} \rangle,
\eeq 
where $n\ge0$, $i_1,\dots,i_n\ge1$, and 
$dM := \prod_{1\leq i\leq N} d M_{ii} \prod_{1\leq i<j\leq N} d{\rm Re} M_{ij} d{\rm Im}M_{ij}$. 
They are polynomials of~$N$ and can be computed by the Wick rule. 
By $\langle \tr M^{i_1} \dots \tr M^{i_n} \rangle_c$ we denote the connected ones. 
Namely, when applying the Wick rule to the computation of~\eqref{Gaussianintegral} we keep summation over connected Feynman diagrams only.
According to~\cite{BIZ80, DFGZJ95, HZ86, O00, OP09}, 
\beq
\langle \tr M^{i_1} \dots \tr M^{i_n} \rangle_c = \sum_{g=0}^{[\frac{|i|}{4}+\frac12-\frac{n}2]}  {\rm Map}_{g}(i_1,\dots,i_n) N^{2-2g+\frac{|i|}{2}-n},
\eeq
where $|i|:=i_1+\dots+i_n$, and
${\rm Map}_{g}(i_1,\dots,i_n)$ denotes the number of maps on a genus~$g$ orientable surface 
with $n$ cells marked by $1,\dots,n$ of sizes $i_1,\dots,i_n$ and with a choice of a vertex at the boundary of each cell.
We call ${\rm Map}_{g}(i_1,\dots,i_n)$ an {\it $n$-point connected GUE correlator of genus~$g$}.
Note that $\langle \tr M^{i_1} \dots \tr M^{i_n} \rangle_c$ vanishes unless $|i|$ is even.
The relation to Toda/Volterra integrability will be reviewed in Section~\ref{section2}.

For fixed $g,n$, Okounkov~\cite{O00} considered 
 the asymptotics for ${\rm Map}_{g}(i_1,\dots,i_n)$ when $i_1,\dots,i_n$ are all large 
in the way that 
\beq\label{toinf}
i_a \sim \kappa x_a, \quad \kappa \to \infty.
\eeq
For fixed $g,n\ge0$ satisfying $2g-2+n>0$, 
by comparing with the Kontsevich main identity~\eqref{Kontsmi}, 
 the following remarkable formula was proved in~\cite{O00} (cf.~\cite{OP09}):
\beq\label{Ok}
 \frac{2^{2g-3+\frac{3n}2} \pi^{\frac{n}2}}{\sqrt{x_1\cdots x_n}} \frac{{\rm Map}_g(i_1,\dots,i_n)}{2^{|i|}\kappa^{3g-3+\frac{3n}{2}}} \to Q_g(x_1,\dots,x_n), 
 \quad \mbox{as \eqref{toinf}, for $|i|$ being even},
\eeq
where 
$Q_g(x_1,\dots,x_n)$ denotes the Witten $n$-point function in genus~$g$, i.e.,
\beq
Q_g(x_1,\dots,x_n) := \sum_{d_1,\dots,d_n\ge0} \langle \tau_{d_1} \dots \tau_{d_n} \rangle_g \, x_1^{d_1} \cdots x_n^{d_n}. 
\eeq
By using~\eqref{Ok} and Toda integrability we will give a new proof of Witten's conjecture.

\smallskip 

\noindent {\bf Organization of the paper.}  In Section~\ref{section2} we review Toda lattice and the GUE partition function. 
In Section~\ref{section3} we give a short proof of Witten's conjecture based on~\eqref{Ok}.

\smallskip

\noindent {\bf Acknowledgements.}
I would like to thank Alessandro Giacchetto for stimulating lectures given 
 at MATRIX, Australia. Part of the work was done during my visit at ICTP, Italy; I thank ICTP for 
 warm hospitality. I also thank Don Zagier for several helpful suggestions that improve a lot the 
 presentation of the paper. The work is supported by NSFC 12371254 and CAS YSBR-032.

\section{Review of Toda lattice hierarchy and GUE partition function}\label{section2}
In this section we review the Toda lattice hierarchy and its tau-functions 
by means of the matrix-resolvent method (cf.~\cite{BDY16, DY17, DYZ21, Y20}), and review the Toda integrability for the 
 partition function of GUE correlators.

\subsection{Tau-function for the Toda lattice hierarchy}
Let $\A$ be the polynomial ring:
\beq
\A := \ZZ [ v_0,w_0,v_{\pm 1},  w_{\pm 1}, v_{\pm 2}, w_{\pm 2}, \cdots ].
\eeq
Define the shift operator $\Lambda:\A \rightarrow \A$ as the linear operator satisfying 
\beq
\Lambda (v_i) = v_{i+1}, \quad  \Lambda (w_i) = w_{i+1} ,  
\quad \Lambda (f g) = \Lambda (f) \Lambda (g),\qquad \forall \, i\in \ZZ,\, f, g\in \A.
\eeq
For an operator of the form
$P = \sum_{m\in \ZZ} P_m \Lambda^m $,
with $P_m \in \A$,
denote $P_+:=\sum_{m\geq 0} P_m \Lambda^m$ and $P_-:=\sum_{m< 0} P_m \Lambda^m$.
Let 
$L:=\Lambda + v_0 + w_0  \Lambda^{-1}$, where $\Lambda^{-1}$ denotes the inverse operator of~$\Lambda$.
The {\it abstract Toda lattice hierarchy} is a sequence of derivations $D_i: \A \to \A$, $i\ge1$, defined by
\beq\label{todaderiv}
D_i (L) = [(L^i)_+,L]
\eeq
and by requiring $D_i$ commute with~$\Lambda$. The operator $L$ is called the {\it Lax operator}. 
It was proved in~\cite{F74, Mana74} that the derivations $D_i$, $i\ge1$, all commute.

The Lax operator $L$ can be written in the matrix form
\beq
\mathcal{L} =  \Lambda + U(\lambda) , \quad U(\lambda)= 
\begin{pmatrix} v_0 -\lambda & w_0 \\ -1 & 0 \end{pmatrix}.
\eeq
\begin{lemma}[\cite{DY17, Y20}] \label{lemmaone} There exists a unique $2\times 2$ matrix series 
\beq
R(\lambda)=
\begin{pmatrix} 1 & 0\\0 & 0\end{pmatrix}
+{\rm O}\bigl(\lambda^{-1}\bigr) \in {\rm Mat} \bigl(2, \mathcal{A}[[\lambda^{-1}]]\bigr)
\eeq
satisfying the equation
\beq\label{eqres}
\Lambda(R(\lambda)) U(\lambda)-U(\lambda) R(\lambda)=0
\eeq
along with the normalization conditions
\beq\label{normres}
\tr \,R(\lambda)=1, \quad \det R(\lambda)=0.
\eeq
\end{lemma}
The unique series $R(\lambda)$ in the above lemma is called the {\it basic matrix resolvent}.

If we think of~$v_0$, $w_0$ as two functions $v=v(x;\e)$, $w=w(x;\e)$, the operator $\Lambda$ as $e^{\e\p_x}$, and 
$v_m$, $w_m$ as $v(x+m\e;\e)$, $w(x+m\e;\e)$, with $\e$ being a parameter, then 
the derivations~$D_i$ lead to a hierarchy of 
differential-difference equations:
\begin{align}
& \e\frac{\p v}{\p s_i} =  D_i(v_0) , \quad 
\e\frac{\p w}{\p s_i} = D_i(w_0) , \qquad i\ge1, \label{Todauw} 
\end{align}
which are the {\it Toda lattice hierarchy}.  
The $i=1$ equation reads
\begin{align}
& \e\frac{\p v}{\p s_1} = (\Lambda -1) (w) , \quad \e\frac{\p w}{\p s_1}  = w \cdot (v-\Lambda^{-1}(v)) , \label{Todaequw} 
\end{align}
which is the Toda lattice equation. The second flow of~\eqref{Todauw} reads
\begin{align}
& \e\frac{\p v}{\p s_2} = (\Lambda -1) \bigl( w \cdot (v+\Lambda^{-1}(v))\bigr) , \label{Todauw2-v} \\
& \e\frac{\p w}{\p s_2}  =  w \cdot (\Lambda(w) -\Lambda^{-1}(w) + v^2- (\Lambda^{-1}(v))^2). \label{Todauw2-w} 
\end{align}

\begin{lemma}[\cite{DY17, Y20}]\label{lemmatwo} 
For an arbitrary solution $(v(x,{\bf s};\e), w(x,{\bf s};\e))$ to the Toda lattice hierarchy~\eqref{Todauw}, 
there exists a function $\tau(x,{\bf s};\e)$ such that
\begin{align}
&\e^2\sum_{i,j\ge1} \frac1{\lambda^{i+1}\mu^{j+1}}\frac{\partial^2\log \tau(x,{\bf s};\e)}{\partial s_{i} \partial s_{j}}
= \frac{\tr \, (R(\lambda;x,{\bf s};\e) R(\mu;x,{\bf s};\e))}{(\lambda-\mu)^2} - \frac1{(\lambda-\mu)^2},
\label{taun1}\\
& \e\sum_{i\ge1}  \frac1{\lambda^{i+1}} \frac{\p}{\p s_{i}} \biggl( \log \frac{\tau(x+\e,{\bf s};\e)}{\tau(x,{\bf s};\e)}\biggr) = (R(\lambda;x,{\bf s};\e))_{21}, 
\label{taun2} \\
&
\frac{\tau(x+\e,{\bf s};\e) \tau(x-\e,{\bf s};\e)}{\tau(x,{\bf s};\e)^2}=w(x,{\bf s};\e).
\label{taun3}
\end{align}
Here, $R(\lambda;x,{\bf s};\e)=R(\lambda)|_{v_m\mapsto v(x+m\e, {\bf s};\e), \, w_m\mapsto w(x+m\e,{\bf s};\e), \, m\in\ZZ}$.
\end{lemma}
The function $\tau(x,{\bf s};\e)$ in the above lemma is called 
a {\it tau-function of the solution $(v(x,{\bf s};\e), w(x,{\bf s};\e))$} 
to the Toda lattice hierarchy \cite{DY17}.

\subsection{The GUE partition function}
Like in e.g.~\cite{DY17}, it is convenient to introduce the 't Hooft coupling constant $x=N\e$ and 
define the {\it GUE free energy $\F^{\rm G}(x,{\bf s};\e)$} as follows:
\begin{align}
 \F^{\rm G}(x,{\bf s};\e) = & \sum_{g\ge0}
 \sum_{n\geq 1} \sum_{\substack{i_1,\dots,i_n \geq 1 \\ 2 - 2g - n + \frac{|i|}2\ge1}} \frac{{\rm Map}_g(i_1,\dots,i_n)}{n!} 
 s_{i_1} \cdots s_{i_n} \e^{2g-2} x^{2-2g - n + \frac{|i|}2} \label{GUEfree}\\
&+\frac{x^2}{2\e^2}\Bigl(\log x-\frac32\Bigr) 
- \frac{\log x}{12} + \zeta'(-1) + \sum_{g\geq2} \frac{\e^{2g-2} B_{2g}}{4g(g-1)x^{2g-2}}, \nn
\end{align}
where $\zeta(s)$ denotes the Riemann zeta function and $B_m$ the $m$th Bernoulli number.
We call the exponential $\exp \F^{\rm G}(x,{\bf s};\e)=:Z^{\rm G}(x,{\bf s};\e)$ the {\it partition function of GUE 
correlators} or the {\it GUE partition function}, which is well known to 
satisfy the following equations:
\begin{align}
& \sum_{i\geq1} i \bigl(s_i-\frac12\delta_{i,2}\bigr) \frac{\p Z^{\rm G}(x,{\bf s};\e)}{\p s_{i-1}} + \frac{xs_1}{\e^2} Z^{\rm G}(x,{\bf s};\e) = 0, \label{stringG} \\
& \sum_{i\geq1} i \bigl(s_i-\frac12\delta_{i,2}\bigr) \frac{\p Z^{\rm G}(x,{\bf s};\e)}{\p s_{i}} + \frac{x^2}{\e^2} Z^{\rm G}(x,{\bf s};\e) = 0. \label{scalingG}
\end{align}
Equation~\eqref{stringG} is called the {\it string equation} for $Z^{\rm G}(x,{\bf s};\e)$.
Equation~\eqref{scalingG} implies that 
\beq\label{dilationG}
{\rm Map}_{g}(2,i_1,\dots,i_n) = |i| \cdot {\rm Map}_{g}(i_1,\dots,i_n) + \delta_{g,0} \delta_{n,0}.
\eeq

Define two power series $v^{\rm G}(x,{\bf s};\e)$ and $w^{\rm G}(x,{\bf s};\e)$ of~${\bf s}$ by
\begin{align}
&v^{\rm G}(x,{\bf s};\e) = \e (\Lambda-1) \frac{\p \F^{\rm G}(x,{\bf s};\e)}{\p s_1},  \label{defVWintro1}\\
&w^{\rm G}(x,{\bf s};\e) 
= e^{\F^{\rm G}(x+\e, \, {\bf s}; \, \e) + \F^{\rm G}(x-\e, \, {\bf s}; \, \e) - 2 \F^{\rm G}(x, \, {\bf s}; \, \e)}. \label{defVWintro2}
\end{align}
It is well known (see e.g. \cite{AvM95, DY17, GMMMO91, M94}) that 
$(v^{\rm G}(x,{\bf s}; \e), w^{\rm G}(x,{\bf s};\e))$ is a solution to the Toda lattice hierarchy~\eqref{Todauw}
 and that $Z^{\rm G}(x,{\bf s};\e)$ 
is a tau-function of this solution. 
We recall that (see e.g.~\cite{DY17}) the functions 
$v^{\rm G}(x,{\bf s}; \e), w^{\rm G}(x,{\bf s};\e)$ can be uniquely determined by 
the Toda lattice hierarchy~\eqref{Todauw} together with the initial data
\beq\label{inigue}
v^{\rm G}(x, {\bf 0}; \e) = 0, \quad w^{\rm G}(x,{\bf 0};\e)=x.
\eeq

Recall that the {\it even GUE free energy} $\F^{\rm eG}(x,{\bf s}_{\rm even};\e)$ is defined by 
\beq
\F^{\rm eG}(x,{\bf s}_{\rm even};\e)=\F^{\rm G}(x,{\bf s};\e)|_{{\bf s}_{\rm odd}={\bf 0}},
\eeq
where ${\bf s}_{\rm even}=(s_2,s_4,s_6,\dots)$, ${\bf s}_{\rm odd}=(s_1,s_3,s_5,\dots)$. The exponential 
$e^{\F^{\rm eG}(x,{\bf s}_{\rm even};\e)}=:Z^{\rm eG}(x,{\bf s}_{\rm even};\e)$ is called the 
{\it even GUE partition function}, which plays an important role in quantum gravity \cite{GMMMO91, Wi91}. 
By definition we know that  
\beq\label{vvanishi}
v^{\rm G}(x,{\bf s}; \e)|_{{\bf s}_{\rm odd}={\bf 0}} \equiv 0.
\eeq
Denote $w^{\rm eG}(x,{\bf s}_{\rm even};\e)=w^{\rm G}(x,{\bf s}; \e)|_{{\bf s}_{\rm odd}={\bf 0}}$, 
which by definition (see~\eqref{defVWintro2})  satisfies 
\beq\label{weGFeG}
w^{\rm eG}(x,{\bf s}_{\rm even};\e) = 
e^{\F^{\rm eG}(x+\e, \, {\bf s}_{\rm even}; \, \e) + \F^{\rm eG}(x-\e, \, {\bf s}_{\rm even}; \, \e) - 2 \F^{\rm eG}(x, \, {\bf s}_{\rm even}; \, \e)}.
\eeq
Putting ${\bf s}_{\rm odd}={\bf 0}$ on both sides of~\eqref{Todauw2-w} and using~\eqref{vvanishi}, one finds that 
$w^{\rm eG}(x,{\bf s}_{\rm even};\e)$ satisfies 
the Volterra lattice equation (also called the discrete or difference KdV equation):
\beq\label{Volterraeq}
\e\frac{\p w^{\rm eG}(x,{\bf s}_{\rm even};\e)}{\p s_{2}} = 
w^{\rm eG}(x,{\bf s}_{\rm even};\e)(w^{\rm eG}(x+\e,{\bf s}_{\rm even};\e) - w^{\rm eG}(x-\e,{\bf s}_{\rm even};\e)).
\eeq
This was known in~\cite{DY17-2, GMMMO91, Wi91}. Similarly, it is clear (see e.g.~\cite{DLYZ20, DY17-2, GMMMO91, Wi91}) that 
$w^{\rm eG}(x,{\bf s}_{\rm even};\e)$ satisfies
\beq\label{Vlattice}
\e\frac{\p w^{\rm eG}(x,{\bf s}_{\rm even}; \e)}{\p s_{2j}} = [(L_{\rm e}^{2j})_+,L_{\rm e}],\quad j\ge1,
\eeq
which is the {\it Volterra lattice hierarchy}, {\it aka} the {\it discrete KdV hierarchy}. 
Here, the Lax operator has the expression $L_{\rm e}=\Lambda+w^{\rm eG}(x,{\bf s}_{\rm even};\e)\Lambda^{-1}$. 

\section{A new proof of Witten's conjecture}\label{section3}
In this section we give a new proof of Witten's conjecture. 

Let us do some preparations. In Section~\ref{section1}, for $2g-2+n>0$
the Witten $n$-point function in genus~$g$, i.e., $Q_g(x_1,\dots,x_n)$, 
is defined. For the cases with $2g-2+n\le0$, like in e.g.~\cite{AIS21, LX11-1, O02}, it is convenient to 
extend the definition of $Q_g(x_1,\dots,x_n)$ as follows:
\begin{align}\label{exddefBg}
Q_{0}(\emptyset)=Q_{1}(\emptyset)=0, \quad Q_{0}(x_1) = \frac1{x_1^2}, \quad Q_{0,2}(x_1,x_2)=\frac1{x_1+x_2}.
\end{align}
Then, as shown by Liu--Xu~\cite{LX11-1}, equation~\eqref{stringeqint} implies that for any $g\ge0$ and $n\ge1$,
\beq
Q_g(x_1,\dots,x_n,x_{n+1}=0) = (x_1+\dots+x_{n}) Q_g(x_1,\dots,x_{n}).
\eeq
By induction we know that for $g\ge0$, $n\ge1$, $I=\{1,\dots,n\}$ and for any $s\ge0$,
\begin{align}\label{stringQ}
&Q_g(x_I,x_{n+1}=0,\dots,x_{n+s}=0) = |x_I|^s Q_g(x_I).
\end{align}
Here and below, $x_I=(x_i)_{i\in I}$ and $|x_I|=\sum_{i\in I} x_i$.

Under the extended definition of $Q_g(x_1,\dots,x_n)$ we also have the following lemma.
\begin{lemma}
Formula~\eqref{Ok} also holds for $(g,n)=(0,1)$ and $(g,n)=(0,2)$.
\end{lemma} 
\begin{proof}
For $g=0$ and $n=1$, it is well known (see e.g.~\cite{HZ86}) that 
\beq
{\rm Map}_0(2j_1) = \frac1{j_1+1}\binom{2j_1}{j_1},\quad \forall\, j_1\ge1,
\eeq
which by Stirling's formula implies that $\frac{2^{-\frac32}\pi^{\frac12}}{\sqrt{x_1}}\frac{{\rm Map}_0(2j_1)}{2^{2j_1} \kappa^{-\frac{3}2}}$
tends to $\frac1{x_1^2} =  Q_0(x_1)$
as~\eqref{toinf}.

For $g=0$ and $n=2$, it is known from~\cite{HZ86} that 
for any $j_1, j_2\ge1$,
\beq\label{g0n2GUE}
{\rm Map}_0(2j_1,2j_2) = \tbinom{2j_1}{j_1}\tbinom{2j_2}{j_2} \tfrac{j_1 j_2}{j_1+j_2},\quad 
{\rm Map}_0(2j_1-1,2j_2-1) = \tbinom{2j_1-1}{j_1} \tbinom{2j_2-1}{j_2} \tfrac{j_1 j_2}{j_1+j_2-1}.
\eeq
Applying Stirling's formula we see that both 
$\frac{\pi}{\sqrt{x_1 x_2}} \frac{{\rm Map}_0(2j_1,2j_2)}{2^{2j_1+2j_2}}$ and 
$\frac{\pi}{\sqrt{x_1 x_2}} \frac{{\rm Map}_0(2j_1-1,2j_2-1)}{2^{2j_1+2j_2-2}}$ 
tend to $\frac1{x_1+x_2} = Q_0(x_1,x_2)$ 
as~\eqref{toinf}.
The lemma is proved.
\end{proof}
We note that for $(g,n)=(0,0)$, $(1,0)$, 
formula~\eqref{Ok} still holds with the extended definition~\eqref{exddefBg}, but trivially.

The following lemma gives an equivalent description of  
Witten's KdV equation.
\begin{lemma}
The $d=1$ case of Witten's conjecture is equivalent to validity of the following relations: for all $g\ge0$ and $n\ge1$, setting $I=\{1,\dots,n\}$, 
\beq\label{QgrecursionKdV}
(2g+n-1)  |x_I|^2 Q_g(x_I) = \frac{|x_I|^5}{12} Q_{g-1}(x_I) +
\sum_{\substack{g_1,  \, g_2\ge0 \\ g_1+g_2=g}} 
\sum_{\substack{A\neq \emptyset, B\neq \emptyset \\ A\sqcup B=I}} |x_{A}|^2 |x_{B}|^3 Q_{g_1}(x_A) Q_{g_2}(x_B).
\eeq
\end{lemma}
\begin{proof}
In terms of the Witten free energy~$\F({\bf t})$ (see~\eqref{Wittenfreeenergy}), the $d=1$ case of Witten's conjecture reads as follows:
\beq\label{KdVfree}
\frac{\p^3 \F({\bf t})}{\p t_1\p t_0^2} = \frac{\p^2 \F({\bf t})}{\p t_0^2}\frac{\p^3 \F({\bf t})}{\p t_0^3} 
+ \frac{1}{12} \frac{\p^5 \F({\bf t})}{\p t_0^5},
\eeq
which, by comparing coefficients of powers of~${\bf t}$ and under degree-dimension matching (or by dilaton equation), is equivalent to
the following relations: for all $g,n, d_1,\dots,d_n\ge0$,
\beq\label{equivkdv0}
\langle \tau_1\tau_0^2 \tau_{d_I}\rangle_g = \frac{1}{12} \langle \tau_0^5\tau_{d_I}\rangle_{g-1} + 
\sum_{\substack{g_1, \, g_2\ge0 \\ g_1+g_2=g}} \sum_{A\sqcup B=I} \langle \tau_0^2 \tau_{d_A}\rangle_{g_1}  \langle \tau_0^3 \tau_{d_B}\rangle_{g_2}
\eeq
(see also~\cite{FP00, LX11-1,Wi91}). Here, $\tau_{d_A}=\prod_{a \in A} \tau_{d_a}$.
For $n=0$, the equality~\eqref{equivkdv0} holds trivially. 
For $n\ge1$, noticing $\langle\tau_0^3\rangle_0=1$, by~\eqref{dilatoneqint} and~\eqref{stringQ} the equality~\eqref{equivkdv0} simplifies to~\eqref{QgrecursionKdV}.
The lemma is proved.
\end{proof}

\begin{remark}
The relation~\eqref{QgrecursionKdV} is similar to~\cite[(ii) of Proposition~2.1]{LX11-1} (see also~\cite{LX11-2}).
We note that there is another version of Witten's conjecture, proved in~\cite{Ko92}, which says that 
the function $Z({\bf t}):=\exp \F({\bf t})$ is a  
tau-function for the KdV hierarchy. For details about KdV tau-function see e.g.~\cite{BDY16, Dickey03, DYZ21, KS91, Ko92}.
It implies the original version of Witten's conjecture, 
which, together with \eqref{stringeqint} and~\eqref{dilatoneqint}, also
implies back the tau-function version. A particular case of the tau-function version says (cf. e.g.~\cite{BDY16, DYZ21, LX11-2})
\beq\label{bilinear1}
\frac{\p^2 \F({\bf t})}{\p t_1\p t_0} = \frac12 \biggl(\frac{\p^2 \F({\bf t})}{\p t_0^2}\biggr)^2
+ \frac{1}{12} \frac{\p^4 \F({\bf t})}{\p t_0^4}
\eeq
(differentiation with respect to $t_0$ gives~\eqref{KdVfree}).
The following relations were obtained by Liu--Xu~\cite{LX11-2}
from~\eqref{bilinear1}, \eqref{stringeqint}, \eqref{dilatoneqint}: 
\beq\label{LXrecursion}
(2g+n-1)  |x_I| Q_g(x_I) = \frac{|x_I|^4}{12} Q_{g-1}(x_I) +
\sum_{\substack{g_1, \, g_2\ge0 \\ g_1+g_2=g}} 
\sum_{\substack{A\neq \emptyset, B\neq \emptyset \\ A\sqcup B=I}} |x_{A}|^2 |x_{B}|^2 Q_{g_1}(x_A) Q_{g_2}(x_B)
\eeq
(cf.~also~\cite{LX11-1} for the discovery of~\eqref{LXrecursion} again assuming Witten's conjecture).
\end{remark}

The following lemma gives an identity for the even GUE free energy.
\begin{lemma}
The following identity holds for $\F^{\rm eG}=\F^{\rm eG}(x,{\bf s}_{\rm even};\e)$:
\beq\label{idFeg1}
\e \, \frac{\Lambda-1}{\Lambda+1}\biggl(\frac{\p^2\F^{\rm eG}}{\p x \p s_2}\biggr) 
= \biggl(\e \, \frac{\Lambda-1}{\Lambda+1}\biggl(\frac{\p\F^{\rm eG}}{\p s_2}\biggr)\biggr)  \cdot 
\biggl((\Lambda-1)(1-\Lambda^{-1})\biggl(\frac{\p\F^{\rm eG}}{\p x}\biggr)\biggr).
\eeq
\end{lemma}
\begin{proof}
It is known from~e.g.~\cite{DY17} (cf.~\eqref{taun2}) that 
\beq
\e (\Lambda-1)\biggl(\frac{\F^{\rm G}(x,{\bf s};\e)}{\p s_2}\biggr) = (v^{\rm G}(x,{\bf s};\e))^2 + (\Lambda+1)(w^{\rm G}(x,{\bf s};\e)).
\eeq
Taking ${\bf s}_{\rm odd}={\bf 0}$, using~\eqref{vvanishi}, and applying $(\Lambda+1)^{-1}$ to the resulting identity, we get
\beq\label{weGanother}
\e \frac{\Lambda-1}{\Lambda+1}\biggl(\frac{\F^{\rm eG}(x,{\bf s}_{\rm even};\e)}{\p s_2}\biggr) = w^{\rm eG}(x,{\bf s}_{\rm even};\e)
\eeq
(this identity was obtained in~\cite[(3.32)]{DLYZ20}).
Dividing both sides of~\eqref{Volterraeq} by $w^{\rm eG}(x,{\bf s}_{\rm even};\e)$, applying the operator $\frac{\p_x}{\Lambda-1}\circ \frac{\Lambda}{\Lambda+1}$, 
and using \eqref{weGFeG}, \eqref{weGanother}, we obtain~\eqref{idFeg1}.
\end{proof}

We note that by~\eqref{weGFeG} and~\eqref{weGanother} the even GUE free energy $\F^{\rm eG}$ satisfies
\beq\label{tauid1}
\e \frac{\Lambda-1}{\Lambda+1}\biggl(\frac{\p\F^{\rm eG}}{\p s_2}\biggr) 
= e^{(\Lambda-1)(1-\Lambda^{-1})(\F^{\rm eG})}.
\eeq
Let us also give a derivation of~\eqref{idFeg1} just using~\eqref{tauid1}. 
Similarly to~\cite{DYZ17}, applying~$\p_x$ to both sides of~\eqref{tauid1}, we obtain
$$
\e \frac{\Lambda-1}{\Lambda+1}\biggl(\frac{\p^2\F^{\rm eG}}{\p x \p s_2}\biggr) 
= e^{(\Lambda-1)(1-\Lambda^{-1})(\F^{\rm eG})} \cdot 
\biggl((\Lambda-1)(1-\Lambda^{-1})\biggl(\frac{\p\F^{\rm eG}}{\p x}\biggr)\biggr)
$$
and, replacing the exponential term of the right-hand side using~\eqref{tauid1}, we get~\eqref{idFeg1}.

We are ready to give a new proof of the Witten--Kontsevich theorem.

\begin{theorem}[Kontsevich~\cite{Ko92}] \label{thm1} 
Witten's conjecture holds.
\end{theorem}
\begin{proof}
We will prove the $d=1$ case, which, as mentioned in Section~\ref{section1}, together with  \eqref{stringeqint}, \eqref{dilatoneqint}
for psi-class intersection numbers implies all cases.

Let $\F^{\rm norm}(x,{\bf s}_{\rm even};\e):=\F^{\rm eG}(x,{\bf s}_{\rm even};\e)-\F^{\rm eG}(x,{\bf 0};\e)$ be
the normalized even GUE free energy. In terms of $\F^{\rm norm}(x,{\bf s}_{\rm even};\e)$,
 identity~\eqref{idFeg1} reads
\beq\label{idFeg1norm}
\e \, \frac{\Lambda-1}{\Lambda+1}\biggl(\frac{\p^2\F^{\rm norm}}{\p x \p s_2}\biggr) 
= \biggl(\e \, \frac{\Lambda-1}{\Lambda+1}\biggl(\frac{\p\F^{\rm norm}}{\p s_2}\biggr)\biggr)  \cdot 
\biggl(\frac1x+(\Lambda-1)(1-\Lambda^{-1})\biggl(\frac{\p\F^{\rm norm}}{\p x}\biggr)\biggr).
\eeq
Noting that 
\beq
\e\frac{\Lambda-1}{\Lambda+1} = \sum_{g\ge0} \e^{2g+2} \frac{2^{2g+3}-2}{(2g+2)!} B_{2g+2}\p_x^{2g+1}, 
\eeq
and for fixed $h\ge0$ comparing coefficients of~$\e^{2h}$ on both sides of~\eqref{idFeg1norm}, we obtain
\begin{align}
&\sum_{\substack{g,\,g'\ge0\\  g+g'=h}} \frac{(2^{2g+3}-2)B_{2g+2}}{(2g+2)!} \p_x^{2g+2}\biggl(\frac{\p\F^{\rm norm}_{g'}}{\p s_2}\biggr) 
\label{identgenush}\\
&=\sum_{\substack{g_1, \, g_2, \, g_1', \, g_2' \ge0 \\  g_1+g_2+g_1'+g_2'=h}} 
\frac{(2^{2 g_1+3}-2)B_{2 g_1+2} }{(2 g_1+2)!}   \p_x^{2g_1+1}\biggl(\frac{\p\F^{\rm norm}_{g_1'}}{\p s_2}\biggr) 
\biggl(\frac{\delta_{g_2,0} \delta_{g_2',0}}{x} + 2\frac{\p_x^{2g_2+3}(\F^{\rm norm}_{g_2'})}{(2g_2+2)!} 
\biggr),\nn
\end{align}
where $\F^{\rm norm}_{g}=\F^{\rm norm}_{g}(x,{\bf s}_{\rm even}):=[\e^{2g-2}] \F^{\rm norm}(x,{\bf s}_{\rm even};\e)$, $g\ge0$.
Successively taking derivatives with respect to $s_{2j_1},\dots,s_{2j_n}$ on both sides of~\eqref{identgenush}, we get
\begin{align}
&\sum_{\substack{g, \, g'\ge0 \\ g+g'=h}} \frac{(2^{2g+3}-2)B_{2g+2}}{(2g+2)!} \p_x^{2g+2} 
\biggl(\frac{\p^{1+n}\F^{\rm norm}_{g'}}{\p s_2 \p s_{2j_I}}\biggr) \nn\\
&=\sum_{A\sqcup B=I}\sum_{\substack{g_1,\,g_2,\,g_1',\,g_2' \ge0 \\  g_1+g_2+g_1'+g_2'=h}} 
\frac{(2^{2 g_1+3}-2)B_{2 g_1+2} }{(2 g_1+2)!}  \p_x^{2g_1+1}\biggl(\frac{\p^{1+|A|}\F^{\rm norm}_{g_1'}}{\p s_2 \p s_{2j_{A}}}\biggr) \nn\\
&\qquad \times \biggl( \delta_{B,\emptyset}\frac{\delta_{g_2,0} \delta_{g_2',0}}{x}  + \frac2{(2g_2+2)!} 
 \p_x^{2g_2+3}\biggl(\frac{\p^{|B|} \F^{\rm norm}_{g_2'}}{\p s_{2j_{B}}}\biggr)\biggr),\nn
\end{align}
where $I=\{1,\dots,n\}$, and for a subset $A\subset I$, $|A|$ denotes the cardinality of~$A$ and $2j_{A}=(2j_{a})_{a\in A}$.
Taking ${\bf s}_{\rm even}={\bf 0}$ gives
\begin{align}
&\sum_{\substack{g,\,g'\ge0 \\ g+g'=h}} \frac{(2^{2g+3}-2)B_{2g+2}}{(2g+2)!} \p_x^{2g+2} 
\bigl({\rm Map}_{g'}(2,2j_I)x^{2-2g' - n + |j|} \bigr) \nn\\
&=\sum_{A\sqcup B=I}\sum_{\substack{g_1,\,g_2,\,g_1',\,g_2' \ge0 \\  g_1+g_2+g_1'+g_2'=h}} 
\frac{(2^{2 g_1+3}-2)B_{2 g_1+2} }{(2 g_1+2)!}   \p_x^{2g_1+1}\bigl({\rm Map}_{g_1'}(2,2j_A)x^{2-2g_1' - |A| + |j_A|} \bigr)  \nn\\
& \quad \times \biggl(\delta_{B,\emptyset}\frac{\delta_{g_2,0} \delta_{g_2',0}}{x} +
 \frac2{(2g_2+2)!}\p_x^{2g_2+3}\bigl({\rm Map}_{g_2'}(2j_B)x^{2-2g_2' - |B| + |j_B|}\bigr)
 \biggr),\nn
\end{align}
where $|j|=j_1+\dots+j_n$ and $|j_A|=\sum_{a\in A} j_a$.
By using~\eqref{dilationG}, we then have, for any fixed $h,n\ge0$, and for all $j_1,\dots,j_n\ge1$,
\begin{align}
&\sum_{\substack{g,\,g'\ge0 \\ g+g'=h}} (2^{2g+3}-2)B_{2g+2} \cdot (2|j| {\rm Map}_{g'}(2j_I)+\delta_{g',0}\delta_{n,0})\binom{2-2g' - n + |j|}{2g+2} \label{identgenushnumber}\\
&=\sum_{A\sqcup B=I}\sum_{\substack{g_1, \,g_2, \, g_1', \, g_2' \ge0 \\  g_1+g_2+g_1'+g_2'=h}} 
\frac{(2^{2 g_1+3}-2)B_{2 g_1+2} }{2 g_1+2}  \binom{2-2g_1' - |A| + |j_A|}{2g_1+1}\nn\\
&\times (2|j_A|  {\rm Map}_{g_1'}(2j_A)+\delta_{g_1',0}\delta_{A,\emptyset}) \nn\\
&\times 
\biggl( \delta_{B,\emptyset}\delta_{g_2,0} \delta_{g_2',0} + 2(2g_2+3) \binom{2-2g_2' - |B| + |j_B|}{2g_2+3}{\rm Map}_{g_2'}(2j_B)\biggr).\nn
\end{align}
For $n\ge1$, we have, after a relatively lengthy simplification,  
\begin{align}
&(2^{3}-2)B_{2} \cdot (2h+n-1) (|j|+2-2h - n)(|j|-2h - n) {\rm Map}_{h}(2j_I) \label{eq56}\\
&+(2^{5}-2)B_{4} \cdot |j| {\rm Map}_{h-1}(2j_I) \tfrac{(|j| + 4-2h - n )(|j| + 3-2h - n)(|j|+2-2h - n)(|j|+1-2h - n)}{12}\nn\\
&+\sum_{g=2}^h (2^{2g+3}-2)B_{2g+2} \cdot 2|j| \, {\rm Map}_{h-g}(2j_I)\binom{2-2(h-g) - n + |j|}{2g+2} \nn\\
&=\sum_{\substack{A\neq \emptyset, B\neq \emptyset \\ A\sqcup B=I}}
\sum_{\substack{g_1', \, g_2' \ge0 \\  g_1'+g_2'=h}} 
(2^{3}-2)B_{2} \cdot (2-2g_1' - |A| + |j_A|) |j_A| {\rm Map}_{g_1'}(2j_A) \nn\\
&\times (|j_B|+2-2g_2' - |B|)(|j_B|+1-2g_2' - |B|)(|j_B|-2g_2' - |B|) {\rm Map}_{g_2'}(2j_B) \nn\\
&+\sum_{\substack{A\neq \emptyset, B\neq \emptyset \\ 
A\sqcup B=I}}
\sum_{\substack{g_1, \, g_2, \, g_1', \, g_2' \ge0 \\ g_1+g_2>0\\  g_1+g_2+g_1'+g_2'=h}} 
\frac{(2^{2 g_1+3}-2)B_{2 g_1+2} }{g_1+1}  \binom{2-2g_1' - |A| + |j_A|}{2g_1+1} |j_A|  {\rm Map}_{g_1'}(2j_A) \nn\\
&\times 2(2g_2+3) \binom{2-2g_2' - |B| + |j_B|}{2g_2+3}{\rm Map}_{g_2'}(2j_B) \nn\\
& + (2^{3}-2)B_{2}  \cdot 10 \binom{2-2(h-1) - n + |j|}{5}{\rm Map}_{h-1}(2j_I)\nn\\
& + \sum_{g_2=2}^h 
(2^{3}-2)B_{2}  \cdot 2(2g_2+3) \binom{2-2(h-g_2) - n + |j|}{2g_2+3}{\rm Map}_{h-g_2}(2j_I)\nn\\
&+ \sum_{g_1=1}^h
\frac{(2^{2 g_1+3}-2)B_{2 g_1+2} }{2 g_1+2}  \binom{2-2(h-g_1) - n + |j|}{2g_1+1} 
{\rm Map}_{h-g_1}(2j_I)\cdot 2|j|.\nn
\end{align}
For any fixed $h\ge0$ and fixed $n\ge1$, 
multiplying both sides of~\eqref{eq56} by 
$$\frac{2^{2h-1+\frac{3n}2} \pi^{\frac{n}2}}{\sqrt{x_1\dots x_n}} \frac1{2^{2|j|}\kappa^{3h-1+\frac{3n}2}},$$
recalling $B_2=\frac16, B_4=-\frac1{30}$,   
and taking limits on both sides of the resulting identity as~\eqref{toinf}, i.e.,
$$
2j_a \sim \kappa x_a, \quad \kappa \to \infty,
$$
we obtain~\eqref{QgrecursionKdV}. The theorem is proved.
\end{proof}

\end{document}